\documentclass[a4paper,11pt]{article}
\usepackage[latin1]{inputenc}
\usepackage[T1]{fontenc}
\usepackage{amsmath, amsthm, amssymb, amsfonts, mathtools}
\usepackage{array}
\usepackage{xcolor}
\usepackage{graphicx}
\usepackage{pgfplots}
\usepackage{makecell}

\usepackage{listings}
\usepackage{enumitem}
\usepackage{nicematrix}
\usepackage{comment}
\usepackage{algorithm}
\usepackage[noend]{algpseudocode}
\usepackage{siunitx}
\usepackage{float}
\usepackage{tikz}
\usetikzlibrary {arrows.meta,automata,positioning} 
\usepackage{mathtools,eqparbox}
\usepackage{booktabs}

\newtheorem{theorem}{Theorem}

\newtheorem{lemma}[theorem]{Lemma}
\newtheorem{proposition}[theorem]{Proposition}
\newtheorem{corollary}[theorem]{Corollary}

\theoremstyle{definition}

\newtheorem{definition}[theorem]{Definition}
\theoremstyle{remark}
\newtheorem{remark}[theorem]{Remark}

\newcommand{\F}{\mathbb{F}}
\newcommand{\G}{\mathbf{G}}

\newcommand{\e}{\mathbf{e}}

\DeclareMathOperator{\wt}{wt}

\DeclareMathOperator{\supp}{supp}

\usepackage{hyperref}
\hypersetup{breaklinks=True}
\usepackage{xurl}
\usepackage[labelfont=bf,labelsep=space]{caption}

\title{Information-Set Decoding for Convolutional Codes}
\author{Niklas Gassner, Julia Lieb, Abhinaba Mazumder, and Michael Schaller}

\begin{document}

\maketitle

\abstract{In this paper, we present a framework for generic decoding of convolutional codes, which allows us to
do cryptanalysis of code-based systems that use convolutional codes as public keys. We then apply this framework to information set decoding, study success probabilities and give tools to choose variables. Finally, we use this to attack two cryptosystems based on convolutional codes. In the case of \cite{bolkema2017variations}, our code recovered about 74\% of errors in less than 10 hours each, and in the case of \cite{almeidaBNS21}, we give experimental evidence that 80\% of the errors can be recovered in times corresponding to about 70 bits of operational security, with some instances being significantly lower.}

\maketitle

\section{Introduction}\label{sec1}

Current cryptographic systems rely on the hardness of integer factorisation or the discrete logarithm problem. These problems can be solved efficiently with Shor's algorithm on a quantum computer \cite{shor}.
This led to a new interest in post-quantum cryptography, a part of which is cryptography based on linear codes. The first code-based cryptosystem was proposed in 1978 by McEliece \cite{mceliece} and uses binary Goppa-codes. The proposal is still fundamentally intact and the proposal Classic McEliece \cite{NISTMcEliece} is based on it.
Round 4 of the currently ongoing post-quantum competition of the National Institute of Standards and Technology (NIST) features three code-based public-key submissions, the already mentioned Classic McEliece, BIKE \cite{NISTBike}, and HQC \cite{NISTHQC}, which was recently selected for standardisation \cite{nist2025status}. 

There have been proposals not featured in the NIST post-quantum competition based on convolutional codes, such as \cite{londahl2012new} (attacked in \cite{londahlattack}), \cite{almeidaBNS23}, or \cite{moufek2018new}. The rationale behind this is that the key size is linear in the memory of the convolutional code, while security levels are reliant on the size of the sliding generator or parity-check matrix, leading to an expected exponential increase in time for generic decoding methods and thus, increased security.
For example, in the Viterbi algorithm \cite{viterbi} the decoding complexity is exponential in the degree.
Additionally for the receiver of the ciphertext, decoding a convolutional code can be done sequentially.

In this paper, we provide a framework for generic decoding of convolutional codes (which are not tail-biting) which is similar to sliding window decoding and apply it to information-set decoding.
The framework resembles the idea of sequential decoding, which was introduced in \cite{wozencraft}.
We study success probabilities of several aspects of the algorithm and give tools to choose parameters in the framework.
Finally, we use it to attack the system proposed in \cite{bolkema2017variations} and a set of parameters of the system proposed in \cite{almeidaBNS21} which is based on the system from \cite{almeida2019mceliece}.
Note that there exist updated versions \cite{almeidaBNS23, almeidaBNS24convolutional} of \cite{almeidaBNS21}  with tail-biting convolutional codes, where our attack does not apply.

The paper is organised as follows. 

In Section \ref{sec: basics}, we give an overview of linear codes and especially convolutional codes, where we introduce all the notions and notations necessary for understanding our work.  

In Section \ref{sec: isdbasics}, we cover the basics of information-set decoding, a class of generic decoding methods for linear codes, where we mainly talk about Prange's algorithm.

The remaining sections cover our own contribution.
In Section \ref{sect:isd-convolutional-codes}, we discuss information-set decoding for convolutional codes. We first give a general framework for generic decoding of convolutional codes that reduces the problem to decoding a smaller block code multiple times and then discuss how we apply it to information-set decoding.
We justify our choice of a depth-first algorithm, then discuss the issues posed by low weight codewords and how we adapt the algorithm to address these issues. Many of the components of our algorithm are probability based, so we provide tools that allow us to compute success probabilities for given parameters and thus help us to choose parameters.

Finally, in Section \ref{sec:experiments} we describe how we can apply our framework to McEliece type cryptosystems whose public key is a convolutional code.
In addition, we describe our implementation of the attack and the results of attacking the cryptosystem proposed in \cite{bolkema2017variations} and a parameter set of the cryptosystem proposed in \cite{almeidaBNS21}.
Since computation times were in most cases infeasible for the latter, we only provide an estimate of the computation time necessary for recovering errors for most random seeds.
We also ran the algorithm in full for two cases, where the estimates of the computation time were low enough, and managed to recover the correct error in both cases.

\subsection{Notation and Conventions}

For the convenience of the reader, we summarize some notation that we will use throughout the paper.
\begin{itemize}[label={--}]
    \item $\mathbb{F}_q$ the finite field with $q$ elements,
    \item $\mathcal{C} \subset \mathbb{F}_q[z]^n$ a convolutional code, where $n$ is its length and $k$ its rank as an $\mathbb{F}_q[z]$-submodule of $\mathbb{F}_q[z]^n$,
    \item $C \subseteq \F_q^N$ a block code of length $N$ and dimension $K$,
    \item matrices and vectors are denoted with bold capital respectively small letters (e.g. $\mathbf{G}$ and $\mathbf{v}$), 
    \item $\mathbf{G}(z)$ and $\mathbf{H}(z)$ for a generator matrix and parity-check matrix, respectively, of a convolutional code $\mathcal{C}$,
    \item $I$ an information set.
    
\end{itemize}

\section{Basics of Linear Codes and Convolutional Codes}
\label{sec: basics}

In this section we present definitions and results for linear block and convolutional codes that will be important in later sections of this paper.

\begin{definition}
Let $K ,N \in\mathbb N$ with $K\leq N$. A linear $[N,K]$-\textbf{block code} $C$ is a $K$-dimensional subspace of $\F_q^N$.
Hence, there is a full rank matrix $\G\in\F_q^{K\times N}$ such that $$C=\{\mathbf{c}\in\F_q^N\ |\ \mathbf{c}=\mathbf{m\G}\ \text{for}\ \mathbf{m}\in\F_q^K\}.$$
 $\G$ is called  \textbf{generator matrix}, $N$ \textbf{length} and ${K}/{N}$ \textbf{rate} of $C$.
\end{definition}

While the generator matrix is used for the encoding of a message, for the decoding of a received word one usually uses another matrix, called parity-check matrix, as defined in the following.

\begin{definition}
   Let $C$ be a linear $[N,K]$-block code.
    A full rank matrix $\mathbf{H} \in \F_q^{(N-K)\times N}$ such that
    $$C=\{\mathbf{c}\in\F_q^N\ |\ \mathbf{Hc}^\top=0 \}$$
    is called \textbf{parity-check matrix} of $C$.
\end{definition}

\begin{lemma}\label{gen_par}
    Let $\G\in\F_q^{K\times N}$ be a generator matrix of a linear block code $C$. A matrix $\mathbf{H}\in\F_q^{(N-K)\times N}$ is a parity-check matrix for $C$ if and only if $\mathbf{HG}^\top=0$ and $\mathbf{H}$ is full rank.
\end{lemma}

\begin{definition}
The \textbf{support} of $\mathbf{c}=(c_1,\hdots,c_N)\in\F_q^N$ is defined as 
$${\rm supp}(\mathbf{c}):=\{i \in \{1,\hdots,N\} \, : \, c_i \neq 0\}.$$
The \textbf{(Hamming) weight} of $\mathbf{c}=(c_1,\hdots,c_N)\in\F_q^N$, denoted by $\wt(\mathbf{c})$, is defined as $|{\rm supp}(\mathbf{c})|$, i.e. as the number of nonzero components of the vector $\mathbf{c}$. 
\end{definition}
\begin{definition}
    The \textbf{minimum distance} of a linear block code $C$ is defined as $$d:=\min\{\mathrm{wt}(\mathbf{c}) \: | \: \mathbf{c} \in C \setminus \{\mathbf{0}\}\}.$$
\end{definition}

In the following, we introduce the basics of convolutional codes, which can be understood as
a generalization of linear block codes.
More details about convolutional codes can be found in \cite{bookchapter}.

\begin{definition}
    A \textbf{convolutional code} $\mathcal{C}$ of length $n$ is defined as an $\F_q[z]$-submodule of $\F_q[z]^n$. Let $k$ be the rank of $\mathcal{C}$. Then, $k/n$ is the \textbf{rate} of $\mathcal{C}$ and we call $\mathcal{C}$ an $(n, k)$ convolutional code.\end{definition}
    \begin{remark}
    Note that $\F_q[z]$ is a Principal Ideal Domain, and thus all submodules of $\F_q[z]^n$ are free. Then, it makes sense to speak of the rank $k$ of the submodule.
\end{remark}
\begin{definition}
    There exists a polynomial matrix $\G(z)\in\F_q[z]^{k\times n}$ such that
    \begin{align*}
        \mathcal{C}  = \{ \mathbf{c}(z) = \mathbf{m}(z)\mathbf{G}(z) \ |\ \mathbf{m}(z) \in \F_q[z]^k \}.
    \end{align*}
 $\G(z)$ is called \textbf{generator matrix} of $\mathcal{C}$. If we write $\G(z) = \sum_{i = 0}^{\mu} \G_iz^i$ with $\G_\mu \not = 0$, then $\mu$ is called \textbf{memory} of $\G(z)$. 
 The maximal degree of the full size (i.e. $k \times k$) minors of $\G(z)$ is called is called \textbf{degree} $\delta$ of $\mathcal{C}$.
\end{definition}
Note that a generator matrix is not unique and the memory $\mu$ depends on $\G(z)$, however the degree $\delta$ of $\mathcal{C}$ does not depend on the choice of the generator matrix.

\begin{definition}\label{def:unimodular}
    A polynomial matrix $\mathbf{U}(z) \in \F_q[z]^{k\times k}$ is called \textbf{unimodular} if there exists $\mathbf{V}(z) \in \F_q[z]^{k\times k}$ such that $\mathbf{V}(z)\mathbf{U}(z) = \mathbf{U}(z) \mathbf{V}(z) = I_k$.
\end{definition}

Unimodular matrices are used to determine whether two matrices generate the same code.

\begin{remark}\label{thm:unimod_generator}
    Two full rank matrices $\G(z), \widetilde{\G}(z) \in \F_q[z]^{k\times n}$ are generator matrices of the same code if and only if there exists a unimodular matrix $\mathbf{U}(z) \in \F_q[z]^{k\times k}$ such that 
    \begin{align*}
        \widetilde{\G}(z) = \mathbf{U}(z)\G(z).
    \end{align*}
$\G(z)$ and $\widetilde{\G}(z)$ are then called \textbf{equivalent}.    
\end{remark}

In the next part of this section, we want to consider parity-check matrices for convolutional codes. In contrast to linear block codes, not every convolutional code admits a parity-check matrix.
 
\begin{definition}
    Let $\mathcal{C}$ be an $(n, k)$ convolutional code. A full row rank polynomial matrix $\mathbf{H}(z) \in \F_q[z]^{(n-k)\times n}$ is called \textbf{parity-check matrix} of $\mathcal{C}$, if
    \begin{align*}
        \mathcal{C} = \{ \mathbf{c}(z) \in \F_q[z]^n \ |\  \mathbf{H}(z)\mathbf{c}(z)^\top = 0 \}.
    \end{align*}
    If such a parity-check matrix exists for the code $\mathcal{C}$, then $\mathcal{C}$ is called \textbf{non-catastrophic}.
\end{definition}

In the following, we will describe how to decide whether a convolutional code possesses a parity-check matrix and how to calculate it if it exists.

\begin{definition}
    For $k\leq n$, $\G(z)\in \F_q[z]^{k\times n}$ is called \textbf{left prime} if in all factorisations $\G(z) = \mathbf{L}(z)\widehat{\G}(z)$, with $\mathbf{L}(z) \in \F_q[z]^{k\times k}$ and $\widehat{\G}(z) \in \F_q[z]^{k \times n}$, the left factor $\mathbf{L}(z)$ is unimodular.
\end{definition}

\begin{definition}
    Let $\G(z)\in \F_q[z]^{k\times n}$ with $k\leq n$ be full (row) rank. Then, there exists a unimodular matrix $\mathbf{U}(z) \in \F_q[z]^{n \times n}$ such that
    \begin{align}\label{herm}
        \G_{rH}(z)= \G(z)\mathbf{U}(z) 
            = \begin{pmatrix}
                h_{11}(z) & 0 & \hdots & 0 & 0 & \hdots & 0 \\
                \vdots & \ddots  & \ddots  & \vdots & \vdots &  & \vdots\\
                \vdots &  & \ddots &  0 & \vdots &  & \vdots \\
                h_{k1}(z) & \dots & \hdots & h_{kk}(z) & 0 & \hdots & 0 \\
            \end{pmatrix}
    \end{align}
    where $h_{ii}(z)$ are monic for $i = 1, \dots , k$ and $\deg(h_{ii}) > \deg(h_{ij})$ for each $j< i$. The matrix $\G_{rH}(z)$ is called \textbf{row Hermite form} of $\G(z)$.
\end{definition}

\begin{theorem}\cite{bookchapter}\label{thm:charact_leftprime}
    Consider $\G(z)\in \F_q[z]^{k\times n}$ with $k\leq n$. The following statements are equivalent:
    \begin{enumerate}
        \item $\G(z)$ is left prime.
        \item The row Hermite form of $\G(z)$ is $\begin{pmatrix} I_k & \mathbf{0} \end{pmatrix}$.
        \item There exists $\mathbf{M}(z) \in \F_q[z]^{n\times k}$ such that $\G(z)\mathbf{M}(z) = I_k$.
        \item $\G(z)$ can be completed to a unimodular matrix, i.e. there exists $\mathbf{E}(z)\in\F_q[z]^{(n-k)\times n}$ such that $\begin{pmatrix}\mathbf{G}(z) \\ \mathbf{E}(z) \end{pmatrix}$ is unimodular.
        \item $rk(\G(\lambda)) = k$ for all $\lambda \in \Bar{\F}_q$, where $\Bar{\F}_q$ denotes the algebraic closure of the field $\F_q$.
    \end{enumerate}
\end{theorem}

\begin{definition}
    A polynomial matrix $\G(z) \in \F_q[z]^{k \times n}$ is said to be delay-free if $\G(0)$ is full row rank. 
    A convolutional code $\mathcal{C}\subset\F_q[z]^n$ is called \textbf{delay-free} if its generator matrices are delay-free.
\end{definition}

It is easy to see that if one generator matrix of a convolutional code is delay-free, then all its generator matrices are delay-free and hence it makes sense to speak of delay-free convolutional codes. Moreover, note that by Theorem \ref{thm:charact_leftprime} all non-catastrophic convolutional codes are delay-free.

\begin{theorem}\cite{bookchapter}\label{thm:PC_existence}
    Let $\mathcal{C}$ be an $(n,k)$ convolutional code. Then, $\mathcal{C}$ admits a parity-check matrix $\mathbf{H}(z) \in \F_q[z]^{(n-k)\times n}$ if and only if any of the generator matrices of $\mathcal{C}$ is left prime.
\end{theorem}

\begin{corollary}
    Let $\G(z)\in\F_q[z]^{k\times n}$ be a left-prime generator matrix of a convolutional code $\mathcal{C}$. A matrix $\mathbf{H}(z)\in\F_q[z]^{(n-k)\times n}$ is a parity-check matrix for $\mathcal{C}$ if and only if $\mathbf{H}(z)\G(z)^\top=0$ and $\mathbf{H}(z)$ is full rank.
\end{corollary}

Next we describe, how a parity-check matrix for a convolutional code can be calculated if it exists.

First the row Hermite form $\G_{rH}(z)$ of $\G(z)\in\F_q[z]^{k\times n}$ is computed. According to Theorem \ref{thm:charact_leftprime} there exists a parity-check matrix if and only if $\G_{rH}(z)= \begin{pmatrix} I_k &  \mathbf{0} \end{pmatrix}$. If this is true, calculate $\mathbf{V}(z)\in\F_q[z]^{n\times n}$ such that $\G(z)\mathbf{V}(z)=\begin{pmatrix}I_k & \mathbf{0} \end{pmatrix}$. Then, the last $n-k$ rows of $\mathbf{V}(z)^\top$ form a parity-check matrix for the convolutional code with generator matrix $\G(z)$.

If $\G(z)\in\F_q[z]^{k\times n}$ is not left prime, the corresponding convolutional code $\mathcal{C}$ does not posses a parity-check matrix. However, it is still possible to find a left prime polynomial matrix $\mathbf{H}(z)\in \mathbb F_q[z]^{(n-k)\times n}$ that plays the role of a parity check matrix for $\mathcal{C}$ in the sense that  
\[
\mathcal{C} \subsetneqq \ker \mathbf{H}(z).
\]
Since $\G(z)$ is not left prime, one has
$$
\G(z)= \arraycolsep=5pt
\begin{pmatrix}
\mathbf{L}(z) & \mathbf{0}
\end{pmatrix}\mathbf{U}(z)
$$
where 
$\begin{pmatrix}
\mathbf{L}(z) & \mathbf{0}
\end{pmatrix}$ is the row Hermite form of $\G(z)$ and $\mathbf{U}(z)\in\F_q[z]^{n\times n}$ unimodular and $\mathbf{L}(z) \in \mathbb F_q[z]^{k\times k}$ is not the identity matrix. It follows that $\G(z) = \mathbf{L}(z) \mathbf{U}_1(z) $, where $\mathbf{U}_1(z)$ consists of the first $k$ rows of $\mathbf{U}(z)$, i.e. $\mathbf{U}_1(z)$ is left prime by Theorem \ref{thm:charact_leftprime}. Hence, the convolutional code $\widetilde{\mathcal{C}}$ with generator matrix $\mathbf{U}_1(z)\in\mathbb{F}_q[z]^{k\times n}$ is non-catastrophic and has a parity-check matrix $\mathbf{H}(z)$. Finally, one has
\[
\mathcal{C} \subsetneqq \widetilde{\mathcal{C}} =\ker \mathbf{H}(z).
\]

In the following we introduce distance measures for convolutional codes.

\begin{definition}
The \textbf{weight} of a polynomial vector $\mathbf{c}(z){=}\displaystyle \sum_{i \in {\mathbb{Z}_{\geq 0}}} \mathbf{c}_i z^i \in \F_q[z]^n$ is defined as
$$
{\rm wt}(\mathbf{c}(z))= \sum_{i \in {\mathbb{Z}_{\geq 0}}} {\rm wt}(\mathbf{c}_i).
$$
\end{definition}

\begin{definition}
    Let $\cal C$ be an $(n,k)$ convolutional code. The \textbf{free distance} of $\cal C$ is defined as
    $$
    d_{free}({\cal C})=\min\{{\rm wt}(\mathbf{c}(z)) \, : \, \mathbf{c}(z) \in {\cal C}\backslash \{\mathbf{0}\}\}.
    $$
\end{definition}

The free distance is a measure for the total number of errors a convolutional code can correct. However, we will carry out the decoding of a convolutional code by splitting the received word and codeword into parts, so-called windows, and decode one window after the other. To this end, we need to introduce so-called sliding matrices.

 Let 
 $$\G(z)=\displaystyle \sum_{i \in {\mathbb{Z}_{\geq 0}}} \mathbf{G}_i z^i \in \F_q[z]^{k \times n}$$  be a generator matrix of a non-catastrophic $(n,k)$ convolutional code $\mathcal{C}$ and 
 $$\mathbf{H}(z)=\displaystyle \sum_{i \in {\mathbb{Z}_{\geq 0}}}  \mathbf{H}_iz^i\in\F_q[z]^{(n-k)\times n}$$ be a parity-check matrix, where $\mathbf{G}_i=\mathbf{0}$ if $i>\deg(\mathbf{G}(z))$ and $\mathbf{H}_i=\mathbf{0}$ if $i>\deg(\mathbf{H}(z))$. For $\gamma \in {\mathbb{Z}_{\geq 0}}$, consider the sliding generator-matrix
$$
\tilde{\mathbf{G}}^\gamma_0:=\begin{pmatrix}
     \mathbf{G}_0 & \mathbf{G}_1 & \cdots & \mathbf{G}_\gamma \\
      & \mathbf{G}_0 & \cdots & \mathbf{G}_{\gamma-1} \\
      & & \ddots & \vdots \\
      & & & \mathbf{G}_0
\end{pmatrix}
$$
and the sliding parity-check matrix
$$\tilde{\mathbf{H}}^\gamma _0:=\begin{pmatrix}
    \mathbf{H}_0 & & \\
    \vdots & \ddots & \\
    \mathbf{H}_\gamma & \cdots & \mathbf{H}_0
\end{pmatrix}.$$
Furthermore, for $i > 0$ we define
$$ \tilde{\mathbf{G}}^\gamma_i:=\begin{pmatrix}
     \mathbf{G}_{i(\gamma+1)} & \mathbf{G}_{i(\gamma+1) +1} & \cdots & \mathbf{G}_{i(\gamma+1)+\gamma} \\
      \mathbf{G}_{i(\gamma+1)-1}& \mathbf{G}_{i(\gamma+1)} & \cdots & \mathbf{G}_{i(\gamma+1)+(\gamma-1)} \\
      \vdots& \vdots & \ddots & \vdots \\
      \mathbf{G}_{i(\gamma+1)-\gamma}& \mathbf{G}_{i(\gamma+1)-(\gamma-1)} & \hdots & \mathbf{G}_{i(\gamma+1)}
\end{pmatrix},$$
where $\mathbf{G}_j = \mathbf{0}$ if $j > \deg(\mathbf{G}(z))$.

\begin{definition}

    Let $\cal C$ be an $(n,k)$ 
    convolutional code.
    For $\gamma \in{\mathbb{Z}_{\geq 0}}$, the \textbf{$\gamma$-th column distance} of $\mathcal{C}$ is defined as
\begin{align*}
d_\gamma^c(\mathcal{C}):&=\min\left\{\sum_{t=0}^\gamma {\rm wt}(\mathbf{c}_t)\ |\ \mathbf{c}(z)\in\mathcal{C}\ \text{and}\ \mathbf{c}_0\neq 0\right\}.
\end{align*}
If $\mathcal{C}$ is delay-free with generator matrix $\G(z)=\displaystyle \sum_{i \in {\mathbb{Z}_{\geq 0}}} \G_i z^i \in \F_q[z]^{k \times n}$, it follows
 $$
d^c_\gamma(\mathcal{C})  =  \min\{{\rm wt}((\mathbf{m}_0, \ldots, \mathbf{m}_\gamma) \tilde{\mathbf{G}}^\gamma_0) \, : \; \mathbf{m}(z) \in {\F_q[z]^k} \mbox{ with } \mathbf{m}_0 \neq 0\}.
    $$

If $\mathcal{C}$ is non-catastrophic with parity-check matrix $$\mathbf{H}(z)=\displaystyle \sum_{i \in {\mathbb{Z}_{\geq 0}}}  \mathbf{H}_iz^i\in\F_q[z]^{(n-k)\times n},$$ one obtains
$$d^c_\gamma(\mathcal{C})  = \min\left\{\sum_{t=0}^\gamma {\rm wt}(\mathbf{c}_t)\ |\ \tilde{\mathbf{H}}^\gamma_0 \begin{pmatrix}\mathbf{c}_0\\ \vdots\\ \mathbf{c}_\gamma\end{pmatrix}=\mathbf{0}\ \text{and}\ \mathbf{c}_0\neq \mathbf{0}\right\}.$$
    \end{definition}

One has that
$$
d^c_0(\mathcal{C}) \leq d^c_1(\mathcal{C}) \leq \cdots \leq \lim_{\gamma \rightarrow\infty} d^c_\gamma(\mathcal{C})\leq d_{free}({\cal C}),
$$
reflecting the fact that the larger we choose the decoding window, the more errors can be corrected inside this window.

We now write the equation $\mathbf{m}(z)\mathbf{G}(z)=\mathbf{c}(z)$ with $\ell:=\deg(\mathbf{m}(z))$ as

 \begin{align}\label{encoding}
 (\mathbf{m}_0, \ldots, \mathbf{m}_{\ell})
    \left(\begin{array}{cccccc}
        \mathbf{G}_0 & \cdots & \mathbf{G}_\mu  &      &        &\\
            & \mathbf{G}_0   & \cdots & \mathbf{G}_\mu &        &\\
            &       & \ddots &       & \ddots &\\
            &       &        & \mathbf{G}_0   & \cdots & \mathbf{G}_\mu\\
    \end{array}\right)=(\mathbf{c}_0, \ldots, \mathbf{c}_{\ell+\mu}).
    \end{align}

Assume we decoded the messages up to a time instant $t-1$, then we can use the following equation:  
 \begin{align*}\label{eq:GdecodSys}
 &(\mathbf{m}_0,\dots, \mathbf{m}_{t-1}\ |\ \mathbf{m}_t, \dots, \mathbf{m}_{t+\gamma})
    \left(\begin{array}{cccccc}
        \mathbf{G}_0 & \cdots &  \mathbf{G}_{t-1}  &  \mathbf{G}_{t}    &    \cdots     & \mathbf{G}_{t+\gamma}\\
                       &      \ddots & \vdots   & \vdots  &  & \vdots\\
            &       &   \mathbf{G}_0     & \mathbf{G}_1   & \cdots & \mathbf{G}_{\gamma+1}\\
            \hline
              &       &        & \mathbf{G}_0   & \cdots & \mathbf{G}_{\gamma}\\
              & & & & \ddots & \vdots\\
              & & & & & \mathbf{G}_0
    \end{array}\right) \\
    &=(\mathbf{c}_0, \ldots, \mathbf{c}_{t+\gamma}).
\end{align*}

and rewrite it as 

\begin{align*}
 &(\mathbf{m}_t, \dots, \mathbf{m}_{t+\gamma})
    \left(\begin{array}{ccc}
                   \mathbf{G}_0   & \cdots & \mathbf{G}_{\gamma}\\
               & \ddots & \vdots\\
               & & \mathbf{G}_0
    \end{array}\right) \\
    &=(\mathbf{c}_t, \ldots ,\mathbf{c}_{t+\gamma})- (\mathbf{m}_0,\dots, \mathbf{m}_{t-1})\begin{pmatrix} \mathbf{G}_{t}    &    \cdots     & \mathbf{G}_{t+\gamma}\\
                      \vdots  &  & \vdots\\
            \mathbf{G}_1   & \cdots & \mathbf{G}_{\gamma+1}\end{pmatrix}.
    \end{align*}
Denote the received word after transmission by $\mathbf{r}(z)=\sum_{i\in{\mathbb{Z}_{\geq 0}}}\mathbf{r}_iz^i\in\F_q[z]^n$.

Then,
$(\mathbf{m}_t, \dots, \mathbf{m}_{t+\gamma})$ can be recovered by decoding $$(\mathbf{r}_t, \ldots, \mathbf{r}_{t+\gamma})- (\mathbf{m}_0,\dots, \mathbf{m}_{t-1})\begin{pmatrix} \mathbf{G}_{t}    &    \cdots     & \mathbf{G}_{t+\gamma}\\
                      \vdots  &  & \vdots\\
            \mathbf{G}_1   & \cdots & \mathbf{G}_{\gamma+1}\end{pmatrix}
   $$ 
   in the block code with generator matrix $\tilde{\mathbf{G}}^\gamma_0$.

This process can be iterated by sliding the decoding window by $\gamma+1$ time steps to $(\mathbf{r}_{t+\gamma+1}, \ldots, \mathbf{r}_{t+2\gamma+1})$ and in this way the whole message $\textbf{m}(z)$ can be recovered with several decoding steps in the block code with generator matrix $\tilde{\mathbf{G}}^\gamma_0$.

\section{Basics of Information-Set Decoding}\label{sec: isdbasics}

Information-set decoding (ISD) is a generic decoding technique for linear codes. In this section we introduce the simplest form of ISD, the Prange algorithm \cite{Prange62}. Several improvements exist, such as \cite{LeeB88}, \cite{Stern88}, \cite{dumer1991minimum}, \cite{Peters10}, \cite{BernsteinLP11}, \cite{MayMT11}, and \cite{BeckerJMM12}.
\begin{definition}
\label{def:informationset}
    Let $C$ be an $[N, K]$ linear block code and  $\mathbf{G}$ be a generator matrix of $C$.
    For any subset $I \subseteq \lbrace 1, \ldots, N \rbrace$ we define $\mathbf{G}_I$ to be the submatrix of $\mathbf{G}$ with columns indexed by $I$.
    Then any subset $I \subseteq \lbrace 1, \ldots, N \rbrace$ of size $K$ such that $\mathbf{G}_I$ is invertible is called an \textbf{information set}.
\end{definition}
The general idea is as follows.
Given an erroneous codeword $\mathbf{c}+\mathbf{e}$, we wish to find an information set $I\subset\{1,\hdots,N\}$ such that $I \cap \supp(\mathbf{e}) = \emptyset$. This implies that $\mathbf{c}$ is uniquely determined by the positions indexed by $I$.
If we have done so, we can easily recover the codeword $\mathbf{c}$ with linear algebra, hence also the error.
In order to find such a set, we try random information sets, do the procedure described above and check if the weight of the recovered error is within specified bounds.
We continue until we have found an error of small enough weight.

More formally, we do the following in the Prange algorithm.
 Let $t$ be an upper bound for the number of errors, $\mathbf{G}$ be a generator matrix of the code, $\mathbf{r}$ the received word, $\mathbf{G}_I$ the submatrix of $\mathbf{G}$ with columns indexed by $I$ and $\mathbf{r}_I$ the vector consisting of the entries of $\mathbf{r}$ indexed by $I$.
The algorithm is as follows.
\begin{enumerate}[label=(\arabic*)]
    \item Pick a random information set $I$.
    \item Solve the equation $\mathbf{m} \mathbf{G}_I = \mathbf{r}_I$ for $\mathbf{m}$.
    \item Set $\mathbf{e} = \mathbf{r} - \mathbf{m} \mathbf{G}$.
    \item If ${\rm wt}(\mathbf{e}) \leq t$, output $\mathbf{e}$, else go back to step $1$.
\end{enumerate}

If $t$ is within the unique decoding radius, the algorithm succeeds if an information set $I$ is found such that $I$ has no intersection with the support of the error $\mathbf{e}$.

Assuming that every subset of size $K$ is an information set (which is only true if the code is maximum distance separable, i.e. the minimum distance is $d = N - K + 1$) and that there is a unique solution, we get for each time we do steps $(1), (2)$ and $(3)$, a probability of $\binom{N-K}{t}/{\binom{N}{t}}$ that $I \cap \supp(e) = \emptyset$.
The expected number of iterations, which we call the \textbf{workfactor}, is the reciprocal of this, denoted by
\[
    \mathrm{WF}_t=\frac{\binom{N}{t}}{\binom{N-K}{t}}.
\]

\section{ISD for Delay-Free Convolutional Codes}
\label{sect:isd-convolutional-codes}

In this section, we study ISD for delay-free convolutional codes (i.e. $\mathbf{G}(0)$ has full rank). For this, we rewrite a sliding generator matrix of $\mathbf{G}(z) = \sum_{i=0}^{\mu} \mathbf{G}_i z^i$ as

\begin{align}
\tilde{\mathbf{G}} = \begin{pmatrix}
\tilde{\mathbf G}^\gamma _0 & \tilde{\mathbf G}^\gamma _1 & \cdots & \tilde{\mathbf{G}}^\gamma _l & & &\\
 & \ddots & \ddots & \vdots & \ddots &  &\\
 & & \tilde{\mathbf{G}}^\gamma _0 & \tilde{\mathbf{G}}^\gamma _1 & \cdots & \tilde{\mathbf{G}}^\gamma _l \\
&  &  & \tilde{\mathbf G}^\gamma _0 &  \ddots & \vdots\\
 &  & &  & \ddots & \tilde{\mathbf{G}}^\gamma _1\\
 & &  & & & \tilde{\mathbf G}^\gamma_0\\
\end{pmatrix},
\end{align}
where $\tilde{\mathbf{G}}^\gamma _i$ has size $k (\gamma+1) \times n (\gamma+1)$.
The parameter $\gamma$ will be chosen in Section \ref{Subsec:choosing_gamma} to optimize the running time.
Note that when $\gamma=0$, we have $\tilde{\mathbf{G}}^\gamma _i = \mathbf{G}_i$, but $\tilde{\mathbf{G}}^\gamma _i$ can consist of several of the $\mathbf{G}_j$'s, for example, for $\gamma=2$, we have
$$\tilde{\mathbf{G}}^2 _0 =
\begin{pmatrix}
\mathbf G_0 & \mathbf G_1 &  \mathbf{G}_2 \\
 & \mathbf G_0 & \mathbf{G}_1 \\
 &  & \mathbf{G}_0 
\end{pmatrix},
$$
and
\[
    \tilde{\mathbf{G}}^2 _1 = 
    \begin{pmatrix}
    \mathbf G_3 & \mathbf G_4 & \mathbf{G}_5 \\
    \mathbf G_2 & \mathbf G_3 & \mathbf{G}_4 \\
    \mathbf G_1 & \mathbf G_2 & \mathbf{G}_3 
    \end{pmatrix}.
\]

We will also write the error vector as $\mathbf{e} = (\tilde{\mathbf{e}}_0, \tilde{\mathbf{e}}_1, \ldots, \tilde{\mathbf{e}}_{s-1})$, where $\tilde{\mathbf{e}}_i$ is of size $n(\gamma+1)$.

\subsection{Generic Decoding of Convolutional Codes}

Given the encrypted message $\mathbf{r}(z) = \mathbf{m}(z)\mathbf{G}(z) + \mathbf{e}(z)$, we aim to find $\mathbf{m}(z)\mathbf{G}(z)$ and consequently $\mathbf{m}(z)$. 

As before, we turn a polynomial vector $\mathbf{v}(z) = \sum_{i=0}^{d} \mathbf{v}_i z^i \in \mathbb{F}_q[z]^n$ into $\mathbf{v} = (\mathbf{v}_0, \mathbf{v}_1, \ldots, \mathbf{v}_{d}) \in \mathbb{F}_q^{nd}$. Let $\tilde{\mathbf{v}}_i$ be the concatenation of several consecutive $\mathbf{v}_i$'s. We get $\mathbf{v} = 
(\tilde{\mathbf{v}}_0, \tilde{\mathbf{v}}_1, \ldots, \tilde{\mathbf{v}}_{s-1}).$

The idea is to iteratively decode the convolutional code with the matrix $\tilde{\G}$ as follows:
\begin{enumerate}
    \item\textbf{Step 0:} Use generic decoding to recover $\tilde{\mathbf{e}}_0$ from $\tilde{\mathbf{r}}_ 0 = \tilde{\mathbf{m}}_0 \tilde{\mathbf{G}}^\gamma _0 + \tilde{\mathbf{e}}_0$ and recover $\tilde{\mathbf{m}}_0$ from $\tilde{\mathbf{m}}_0 \tilde{\mathbf{G}}^\gamma _0$ with linear algebra.
    \item\textbf{Step j for j=1,..., s-1:} \begin{enumerate}
        \item Compute $\tilde{\mathbf{m}}_j \tilde{\mathbf{G}}^\gamma _0 + \tilde{\mathbf{e}}_j = \tilde{\mathbf{r}}_j - \sum_{i=0}^{j-1} \tilde{\mathbf{m}}_i\tilde{\mathbf{G}}^\gamma _{j-i}$.
        \item Recover $\tilde{\mathbf{e}}_j$ and $\tilde{\mathbf{m}}_j$ with generic decoding for the code generated by $\tilde{\mathbf{G}}^\gamma _0$.
    \end{enumerate}
\end{enumerate}

For cryptographic purposes there exists only one error $\mathbf{e}(z)$ with at most a certain Hamming weight such that $\mathbf{r}(z)$ can be written as codeword plus error.
However, we might not find a unique solution $\tilde{\mathbf{e}}_j$ in each block.
So at each step, we produce a list of possible error vectors $( \tilde{\mathbf{e}}_{j,1}, \ldots, \tilde{\mathbf{e}}_{j,n_j} )$ where $\tilde{\mathbf{e}}_{j,\rho} \in \F_q^{N}$ for $\rho \in \{1, \ldots, n_j\}$ are the errors that we recovered using ISD in Step $j$.
We also get from the list of possible errors a list of possible messages $( \tilde{\mathbf{m}}_{j,1}, \ldots, \tilde{\mathbf{m}}_{j,n_j})$.
This gives us a tree, and we aim to find a branch that goes to the bottom of the tree. To find such a branch, we use a depth-first algorithm. In Section \ref{subsec:choice_depth_first}, we will first describe this procedure and then justify the choice of a depth-first algorithm over a breadth-first algorithm.
As a generic decoding method, we use information set decoding, but in principle, any generic decoding method can be used.

Note further that one can also use overlapping parts of the received word for decoding and use a consistency check on the recovered messages to limit the number of possibilities in the search.
For our experiments this did not give huge speedups, but it might be useful if one uses different codes.

\subsection{Using ISD for Convolutional Codes}

We want to use information set decoding (ISD) as a generic method of decoding. Let $t$ be the expected weight of an error block $\tilde{\mathbf{e}}_i$, rounded up to the next integer. While we expect $t$ errors in each block, the chance of each block having at most $t$ errors can be very small.
So we introduce a variable $\varepsilon$ and allow up to $t + \varepsilon$ errors in each block. Note that $\varepsilon$ can be chosen so that an arbitrarily high percentage of the errors satisfies this condition, but increasing $\varepsilon$ increases the number of solutions we obtain at each step and thus increases computational time. We will discuss how to choose $\varepsilon$ in Section \ref{sec: probpoly}.

Note that the first solution we find might not be the correct one, or, if the algorithm is currently running on a ``wrong'' branch, no solution may be found. So, aborting the ISD search after an error is found is not an option. Our solution to this issue is to simply run the algorithm for a fixed amount of iterations and collect all solutions. The number of iterations can be chosen such that it is almost guaranteed that the correct error will be in the list of outputs at every step as we will explain later. Thus, if no solutions are found, we expect to be on a wrong branch.

\subsubsection{The Choice of a Depth-First Algorithm}
\label{subsec:choice_depth_first}

Depth-first algorithms, in general, have lower memory usage than breadth-first algorithms. However, our main reason for choosing a depth-first algorithm is that for our purposes, we expect a depth-first algorithm to run faster than a breadth-first algorithm. The reasons for this we shall explain in this subsection.

Depth-first search works as follows.
We start with a tree and always move down the left-most possible branch.
Once we cannot move down any further we backtrack and then move down the left-most branch we have not explored yet.
Note that we do not have to store what we have explored already if we store in an appropriate way which branches we have yet to explore.

In Figure \ref{fig1}, we labeled the nodes in the tree according to the sequence in which we encounter them in a depth first search.
To be precise we go from $1$ to $2$, $3$, $4$.
Then we backtrack to $3$, from which we go to $5$.
After this we backtrack all the way to $1$ from where we go to $6$ and traverse the remaining graph.

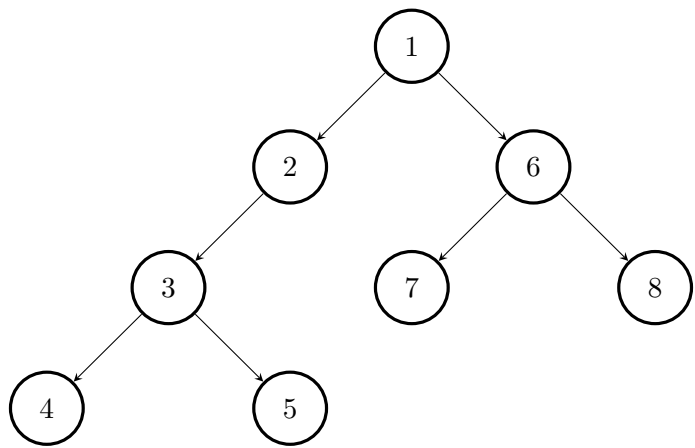
\begin{figure}[ht]
    \centering
    \tikzstyle{node_style}=[shape=circle, draw=black, align=center, scale=1]
    \tikzstyle{edge}=[->,>=stealth, scale=1]
    \tikzstyle{node}=[state, very thick]
    \begin{tikzpicture}[scale=0.8]
        \node[node] (s_8) at (4,-4) {$8$};
        \node[node] (s_7) at (0,-4) {$7$};
        \node[node] (s_6) at (2,-2) {$6$}
            edge[edge]  node[] {} (s_7)
            edge[edge]  node[] {} (s_8);
        \node[node] (s_5) at (-2,-6) {$5$};
        \node[node] (s_4) at (-6,-6) {$4$};
        \node[node] (s_3) at (-4,-4) {$3$}
            edge[edge]  node[] {} (s_4)
            edge[edge]  node[] {} (s_5);
        \node[node] (s_2) at (-2,-2) {$2$}
            edge[edge]  node[] {} (s_3);
        \node[node] (s_1) at (0,0) {$1$}
            edge[edge]  node[] {} (s_2)
            edge[edge]  node[] {} (s_6);
            
    \end{tikzpicture}
    \caption{Depth-First Search.}
    \label{fig1}
\end{figure}

We only worked with the most basic instance of ISD, the Prange algorithm. Recall from Section \ref{sec: isdbasics} that, after a choice of an information set $I$, Prange succeeds in recovering an error if the support of the error is contained in the complement of $I$.
This is more likely to happen for errors of smaller weight, so errors of smaller weight have a higher probability of getting found first by the algorithm than errors of higher weight.

We allow up to $t+\varepsilon$ errors in each window. However, in most cases, the correct error vector will have a smaller weight. Indeed, the average weight in each block is $t$, and the probability of the weight exceeding $t$ by a significant amount is low, as explained in Section \ref{sec:prob}.
This means that the correct error has a tendency to appear at one of the first positions of the output lists of the ISD search at each step. So in most steps, the depth-first algorithm will immediately proceed with the correct error.

Furthermore we heuristically expect a kind of avalanche effect or diffusion, meaning that if we have several wrong messages the error will spread and there will be no close codewords to the received word.
Hence the algorithm should not go too deep in the depth first search for wrong errors and messages.

Our observations from the experiments are in line with these expectations.

\subsubsection{Low Weight Codewords}

Let $\G(z)=\sum \G_iz^i\in\F_q[z]^{k\times n}$ and $\mathbf{H}(z)=\sum \mathbf{H}_iz^i\in\F_q[z]^{(n-k)\times n}$, with $\mathbf{H}_0$ full rank, be a generator matrix and a parity-check matrix, respectively, of the same non-catastrophic convolutional code $\mathcal{C}$. Note that each non-catastrophic convolutional code possesses a left-prime parity-check matrix $\mathbf{H}(z)$, and hence, in particular a parity-check matrix $\mathbf{H}(z)$ such that $\mathbf{H}_0$ is full rank. Then, for $\gamma\in{\mathbb{Z}_{\geq 0}}$, it is an easy consequence of Lemma \ref{gen_par} that $\tilde{\mathbf{H}}^\gamma_0$ is a parity-check matrix for the block code with generator matrix $\tilde{\mathbf{G}}^\gamma_0$.

Note that the block code with generator matrix $\tilde{\mathbf{G}}^\gamma_0$ (like each block code) always possesses a parity-check even if the corresponding convolutional code is catastrophic.
Moreover, it is important to note that, due to the restriction $\mathbf{v}_0\neq \mathbf{0}$, respectively $\mathbf{u}_0\neq \mathbf{0}$, in the definition of column distances, the distance of the block code with generator matrix $\tilde{\mathbf{G}}^\gamma_0$ is upper bounded by $d_\gamma^c$ but can be much smaller.
While a convolutional code $\mathcal{C}$ with generator matrix $\mathbf{G}(z)$ may have a large minimum distance, the linear code generated by the matrix 
$$\tilde{\mathbf{G}}^\gamma _0 = \begin{pmatrix}
\mathbf G_0 & \mathbf G_1 &  \cdots & \mathbf{G}_{\gamma} \\
 & \mathbf G_0 & \ddots & \vdots \\
 &  & \ddots & \mathbf{G}_1 \\
 & & & \mathbf{G}_0
\end{pmatrix}
$$
can have codewords of small weight.
Note that the code generated by $\tilde{\mathbf{G}}^\gamma _0$ contains codewords of the form $(
    \mathbf{0}, \ldots, \mathbf{0} , \mathbf{c} ),$ where $\mathbf{c}$ lies in $\text{rowspan}(\mathbf{G}_0)$.
Note further that if $\mathcal{C}$ is delay-free (that is, $\mathbf{G}_0$ has full rank), every non-zero codeword of $\text{rowspan}(\tilde{\mathbf{G}}^\gamma _0)$ is of the form 
$$(\mathbf{0}, \ldots, \mathbf{0}, \mathbf{c} , \mathbf{r}_1, \ldots, \mathbf{r}_i
)$$ for a non-zero $\mathbf{c} \in \text{rowspan}(\mathbf{G}_0)$. So the minimum distance of the code generated by $\tilde{\mathbf{G}}^\gamma _0$ equals the one of the code generated by $\mathbf{G}_0$.

Actually, if one has a message of the form $\mathbf{m}=(\mathbf{m}_0,\hdots,\mathbf{m}_{\gamma})$ where $i\in\{0,\hdots,\gamma\}$ is minimal such that $\mathbf{m}_i\neq \mathbf{0}$, then the weight of $\mathbf{c}=(\mathbf{c}_0,\hdots,\mathbf{c}_{\gamma})=\mathbf{m}\tilde{\mathbf{G}}_0$ is lower bounded by $d^c_{\gamma-i}(\mathcal{C})$ and $\mathbf{c}_0=\cdots=\mathbf{c}_{i-1}=\mathbf{0}$, $\mathbf{c}_i\neq \mathbf{0}$.

\subsubsection{Undetectable Errors}

Since we are working with Prange's algorithm, errors can't be detected if the support of a codeword is contained in the support of the error. This is formalized in the following proposition, which states that an information set $I$ of a a linear code $C$ over $\mathbb{F}_q$ must intersect the support of each codeword.

\begin{proposition}
    Let $I$ be an information set for an $[N,K]$-code $C$ and $\mathbf{c} \in C \setminus \{ 0\}$. Then $ \mathrm{supp}(\mathbf{c}) \cap I \neq \emptyset$. 
\end{proposition}

\begin{proof}
    Let $\mathbf{c}$ be in $C \setminus \{ 0 \}$ and $I \subset \{ 1,\ldots, N \}$ a set of size $K$ that has trivial intersection with the support of $\mathbf{c}$.
    Pick a basis of $C$ that contains $\mathbf{c}$ and arrange them in a generator matrix $\mathbf{G}$ such that the last row is $\mathbf{c}$. The last row of the matrix $\mathbf{G}_I$, consisting of the columns indexed by $I$, is a zero row, so $\mathbf{G}_I$ is not invertible. This means that $I$ is not an information set.
\end{proof}

Note that in these cases, the error $\mathbf{e}$ can be written as $\mathbf{e} = \mathbf{e'} + \mathbf{c}$, where $\mathbf{c}$ is a codeword whose support is contained in the support of $\mathbf{e}$ and $\mathbf{e'}$ is an error which can be recovered by the Prange algorithm. So we can ensure that $\mathbf{e}$ is in our list of outputs by pre-computing a list of low-weight codewords (for our experiments codewords of weight up to 2 were sufficient for most cases) and then, for all possible errors $\mathbf{e'}$ in the output of the Prange algorithm and low-weight codewords $\mathbf{c}$, appending $\mathbf{e'} + \mathbf{c}$ at the end of the list of outputs if $\wt( \mathbf{e'} + \mathbf{c} ) \leq t+\varepsilon$. It is important that the new solutions get appended at the end of the list when working in a setting where it is unlikely to happen that the support of an error contains the support of a codeword.

\subsubsection{The Choice of $\gamma$}
\label{Subsec:choosing_gamma}

Recall that $\gamma+1$ is the number of distinct $\mathbf{G}_i$'s in our matrix $\tilde{\mathbf{G}}_0^\gamma$ which we use for iterative decoding (see beginning of Section \ref{sect:isd-convolutional-codes}). Note that the choice of $\gamma$ can heavily impact the decoding complexity: choosing a large $\gamma$ will increase the ISD cost in each step, while choosing a small $\gamma$ increases the number of steps. This could potentially substantially increase the nodes the algorithm will transverse in the depth-first search. Thus, $\gamma$ needs to be carefully chosen depending on the parameters of the cryptosystem. In fact, if we choose $\gamma$ to be the maximum of the memory of the generator matrix and the degree of the error vector $\mathbf{e}(z)$, then we are in the setting of decoding a single block code. On the other hand, a small $\gamma$ might also not be suitable for decoding. Consider the case
$$\mathbf{G}(z) = (g_1(z), g_2(z)),$$
where $g_1(z), g_2(z) \in \mathbb{F}_2[z]$ are of large degree. If we choose $\gamma = 0$, that is, $\tilde{\mathbf{G}}_0^\gamma$ consists of the constant terms of $g_1(z)$ and $g_2(z)$, then only a small percentage of errors will have at most $1$ error in each block. But if we decode with at most $2$ errors per block, then we get $4$ solutions in each step, at which point we are essentially just brute-forcing all solutions.

We have not found a rigorous way of choosing $\gamma$ optimally and have resorted to experimentation.

\subsection{ISD for a Fixed Amount of Errors Uniformly Distributed Over The Error Vector} \label{sec:prob}

Our goal in this section is to reduce the information-set decoding of a convolutional code to the ISD of equations of the form

$$ \mathbf{m}_i\tilde{\mathbf{G}^\gamma _0} + \tilde{\mathbf{e}_i},$$
or equivalently
$$ \tilde{\mathbf{H}}^{\gamma}_0 \tilde{\mathbf{e}}^T_i=\mathbf{s}_i.$$

The analysis consists of several steps. We first want to compute the probability that blocks of the error contain at most $t+\varepsilon$ non-zero entries for a given $\varepsilon$.

We will assume the following: $\mathcal{C}$ is a delay-free convolutional code, and the error vector $\mathbf{e}$ of length $s N$ is such that the errors are uniformly distributed over its length.

\subsubsection{Exceedence Probability for the Error Distributions}\label{sec: probpoly}

Let $t_e$ be the total error weight, $N$ the length of the code generated by $\tilde{\mathbf{G}}^\gamma _0$, $s$ the number of blocks of size $N$ we need for decoding, $t = \lceil {t_e}/{s} \rceil$, the (rounded) expected number of errors in each size $N$ window, and $\varepsilon$ the tolerance for the error weights.

To give the exact probability of an error having at most weight $t+\varepsilon$ in each window of size $N$, we consider the polynomial$$p(z) = 1 + (q-1)\binom{N}{1} z +(q-1)^2 \binom{N}{2}z^2 + \ldots + (q-1)^{t+\varepsilon}\binom{N}{t+\varepsilon} z^{t + \varepsilon}.$$
We compute
$$q(z) = p(z)^{s}$$
and recover the $t_e$'th coefficient of $q(z),$ which gives us the number of error vectors of weight $t_e$ with at most $t+\varepsilon$ errors in each block of size $N$. Dividing by the total number of weight $t_e$ error vectors, which is
$$(q-1)^{t_e}\binom{s N}{t_e},$$
yields the probability of having at most $t+\varepsilon$ errors in each block.
One can reformulate this in terms of a multivariate hypergeometric distribution.
More information about multivariate hypergeometric distributions can be found in \cite{johnson1997discrete}.
Note that in each block we get a hypergeometric distribution, which we use for an estimate of the tails in the next section. 

\subsubsection{Estimates, Hypergeometric Distribution}
The hypergeometric distribution has an exponentially decaying tail; see \cite{chvatal1979tail}.
Let us consider the tail of a hypergeometric distribution given by
\[
    H(A, B, b, l) = \binom{B}{b}^{-1} \sum_{i = l}^b \binom{A}{i} \binom{B-A}{b-i}.
\]
Define $\alpha$ through the equation
\[
    l = \left(\frac{A}{B} + \alpha\right) b.
\]
where $\alpha$ should be greater than $0$.
One can show as in \cite{chvatal1979tail} that we get for the tail
\[
    H(A, B, b, l) \leq e^{-2\alpha^2 b}.
\]
In the setting above we have $A = N, B = N s, b = t s$.
Hence
\[
    \frac{A}{B} = \frac{1}{s},
\]
and 
\[
    l = t + \alpha t s.
\]

Thus, if $\varepsilon = \alpha t s$, then we get that the probability of having at least $t + \varepsilon$ errors in a block is at most
\[
    e^{-2 \alpha^2 t s} = e^{-2 \alpha \varepsilon}.
\]

For the error weight of the block to exceed $t+\varepsilon$, we set $l=t+\varepsilon+1$, getting $\alpha=(\varepsilon+1)/st.$ Then the probability of the block having error weight more than $t+\varepsilon$ is at most
\[
    e^{-2\alpha^2 t s} = e^{-2\alpha (\varepsilon+1)}.
\]
By the Union bound, we have:
\begin{align*}
    \mathbb{P}[\operatorname{max} \{ \operatorname{wt}(\tilde{\mathbf{e}}_j), j=0, \cdots, s-1\} > t+\varepsilon] &\leq s H(N, N s, t s, t+\varepsilon+1)\\
    &\leq s e^{-2\alpha^2 t s} = s e^{-2\alpha (\varepsilon+1)}
\end{align*}
where $\tilde{\mathbf{e}}_j$ is the $j$-th error coefficient.

\subsection{Estimates on The Number of Solutions When the Amount of Errors Exceeds the Unique Decoding Radius}
Given a (random) full-rank matrix $\mathbf{H} \in \mathbb{F}_q^{(N-K) \times N}$, an error vector $\mathbf{e} \in \mathbb{F}_q^N$ of weight $t$ and $\mathbf{s} \in \mathbb{F}_q^{N-K}$ we want to compute the expected amount of solutions of the equation
\begin{equation}\label{eqn: pceqn}
\mathbf{He}^T = \mathbf{s}\end{equation}
of weight at most $t+\varepsilon$. We care about the conditional expectation
\begin{align}\label{eq:expected solutions}
&\mathbb{E}[ \#\text{solutions with at most } t+\varepsilon \text{ errors} \: | \: \text{at least one such solution exists} ] \nonumber \\
&\approx 1 + \mathbb{E}[ \#\text{solutions with at most } t+\varepsilon \text{ errors}].
\end{align}

Note that every solution to Equation \eqref{eqn: pceqn} lives in 
$$S := \mathbf{e} + \mathrm{ker}(\mathbf{H}),$$ 
which has $q^{K}$ elements. 
Notice that if $\mathbf{H}$ is ``random'', then we expect this set to be uniformly distributed in $\mathbb{F}_q^N$. So we expect the number of vectors of weight $\tilde{t}$ in $S$ to be
$$\frac{(q-1)^{\tilde{t}}}{q^{N-K}} \binom{N}{\tilde{t}}.$$
Summing over all $\tilde{t} \leq t+\varepsilon$ yields the estimate for (\ref{eq:expected solutions}).

\subsubsection{Choosing Maximum Number of Iterations} \label{wf}

The work factor of the Prange algorithm for an $[N,K]-$code with an error vector of weight $t$ is given by
$$
\mathrm{WF}_t=\frac{\binom{N}{t}}{\binom{N-K}{t}}.
$$
Similarly, if we allow $t + \varepsilon$ errors with $\varepsilon \geq 1$ the work factor changes to
$$
\mathrm{WF}_{t+\varepsilon}=\frac{\binom{N}{t + \varepsilon}}{\binom{N-K}{t + \varepsilon}}.
$$
The quotient of the two work factors simplifies to 
$$
\frac{\mathrm{WF}_{t+\varepsilon}}{\mathrm{WF}_t} = \displaystyle \prod_{j=0}^{\varepsilon-1} \frac{N-t-j}{N-K-t-j} \approx \left({\frac{N-t}{N-K-t}}\right)^\varepsilon.
$$

Note furthermore that while $\mathrm{WF}_{t+\varepsilon}$ is the expected number of iterations to find a solution of weight $t+\varepsilon$, it is unlikely that we will find the correct error at each step with this number of iterations. To ensure that the correct error is in the list of outputs of our ISD algorithm, we run it for a fixed number of iterations, chosen to ensure this.
To make this choice, note that the chance of finding a fixed detectable error of weight $t+\varepsilon$ in one iteration of Prange's algorithm is given by
$$\text{WF}_{t+\varepsilon}^{-1} = \frac{\binom{N-K}{t+\varepsilon}}{\binom{N}{t+\varepsilon}}.$$

So, the probability of finding a given error of weight $t+\varepsilon$ at all $s$ steps with $W$ ISD iterations at each step is given by

$$ \left(1-\left(1-\frac{\binom{N-K}{t+\varepsilon}}{\binom{N}{t+\varepsilon}}\right)^{W}\right)^s.$$

Note that this is a lower bound for the success probability of recovering an error vector $\mathbf{e} = (\tilde{\mathbf{e}}_0, \tilde{\mathbf{e}}_1, \ldots, \tilde{\mathbf{e}}_{s-1})$ with $\mathrm{wt}(\tilde{\mathbf{e}}_i) \leq t+ \varepsilon$ for all $i=1,\ldots, s-1$ since the weight of an $\tilde{\mathbf{e}}_i$ can be smaller than $t+\varepsilon$, which increases the probability of it getting recovered.

\subsubsection{Choice of $\varepsilon$}\label{sec: epsilon}
In this subsection we summarize the impact of the choice of $\varepsilon$.
Increasing $\varepsilon$ makes it more likely that errors satisfy the weight condition on each block and thus allows to recover more errors and messages. However, it also increases the work factor for ISD and the number of solutions found in each step, which can heavily increase the computational time. 

All of these factors need to be carefully considered when choosing $\varepsilon$.

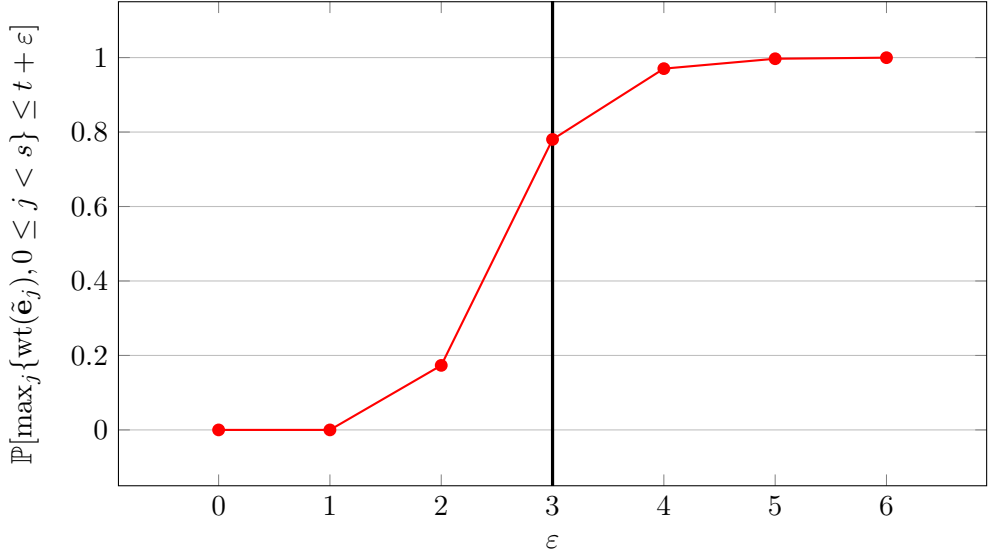
\begin{figure}
  \centering
\begin{tikzpicture}
    \begin{axis}[
        title={Probability that every block has weight at most $t+\varepsilon$ vs. $\varepsilon$},
        xlabel={$\varepsilon$},
        ylabel={$ \mathbb{P}[\operatorname{max} \{ \operatorname{wt}(\tilde{\mathbf{e}}_j), 0 \leq j <s\} \leq t+\varepsilon]$},
        width=13cm,
        height=8cm,
        symbolic x coords={0,
1,
2,
3,
4,
5,
6},
        xtick=data,
        ymin=0, ymax=1,
        legend pos=north east,
        ymajorgrids=true, 
        enlargelimits=0.15,
        legend style={font=\footnotesize},
        ymode=linear,
        every axis plot/.append style={ultra thick}
    ]
    
    \addplot[
        color=red,
        mark=*,
        mark options={fill=red},
        thick,
        nodes near coords={}, 
        every node near coord/.append style={}
    ] coordinates {
        (0, 2.41848981004715e-39)
(1, 0.0000120977983264597)
(2, 0.173289806743814)
(3, 0.780213064712423)
(4, 0.970303327681542)
(5, 0.996893617990054)
(6, 0.999724740527447)

    };
    \addplot[very thick, samples=50, smooth,domain=0:6,black] coordinates {(3,-0.2)(3,1.2)};

    \end{axis}
\end{tikzpicture}

  \caption{Calculating a suitable value of $\varepsilon$ for \cite{bolkema2017variations}, Example 5.11.}
  \label{graphicprob}
\end{figure}

For example, for the attack on \cite{bolkema2017variations} that we will discuss in Section \ref{exp1}, we have the parameters $q=2$, $t_e=140$, $s=167$, $K = 36$ and $N=60.$ This gives us $t=1$, and we get the plot in Figure \ref{graphicprob}, which helps us to choose the value of $\varepsilon$. 
We chose $\varepsilon=3$, as about $78\%$ of errors should have block-weight at most $4$, as can be seen in Figure \ref{graphicprob}, and the work factor is still comparatively small as discussed in Section \ref{wf}, and as can be seen in Figure \ref{graphicwf}. We use this value of $\varepsilon$ for our experiments, as described in Section \ref{exp1}.

\begin{figure}
    \centering
    \begin{tikzpicture}
    \begin{axis}[
        title={Ratio of work factors vs. $\varepsilon$},
        xlabel={$\varepsilon$},
        ylabel={$ {\mathrm{WF}_{t+\varepsilon}}/{\mathrm{WF}_t}$},
        width=13cm,
        height=8cm,
        symbolic x coords={0,
1,
2,
3,
4,
5,
6},
        xtick=data,
        ymin=0, ymax=500,
        legend pos=north east,
        ymajorgrids=true, 
        enlargelimits=0.15,
        legend style={font=\footnotesize},
        ymode=linear,
        every axis plot/.append style={ultra thick}
    ]
    
    \addplot[
        color=red,
        mark=*,
        mark options={fill=red},
        thick,
        nodes near coords={}, 
        every node near coord/.append style={}
    ] coordinates {
        (0, 1.0)
(1, 2.5652173913043477)
(2, 6.762845849802371)
(3, 18.35629587803501)
(4, 51.397628458498026)
(5, 148.7826086956522)
(6, 446.3478260869565)

    };
     \addplot[very thick, samples=50, smooth,domain=0:6,black] coordinates {(3,-100)(3,600)};

    \end{axis}
\end{tikzpicture}

    \caption{${\mathrm{WF}_{t+\varepsilon}}/{\mathrm{WF}_t}$ for Prange on the $[60, 36]-$code, as we vary $\varepsilon$. }
    \label{graphicwf}
\end{figure}
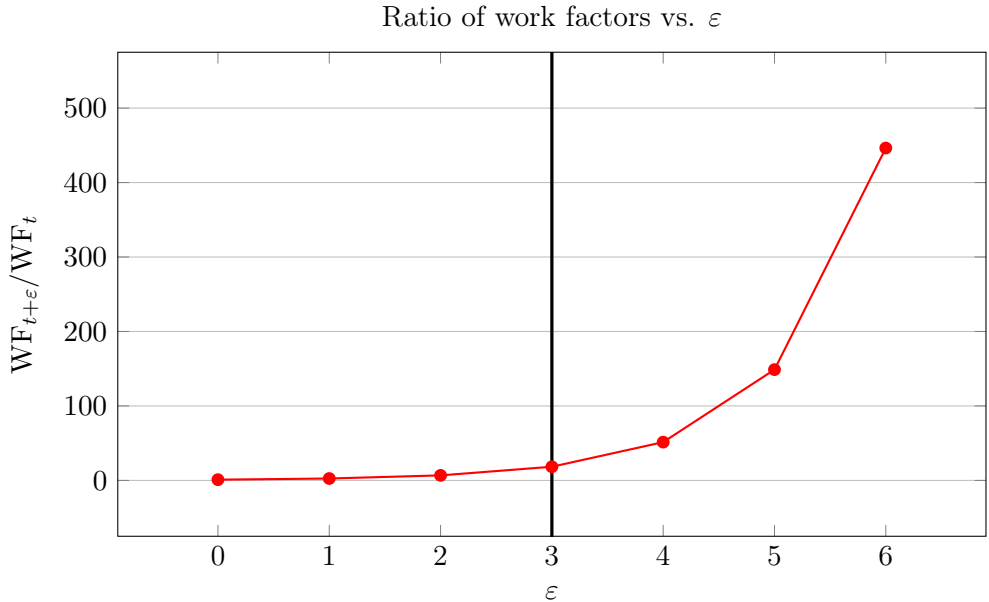

\section{Attacks on McEliece Type Cryptosystems Based on Convolutional Codes}
\label{sec:experiments}

In this section we describe how to attack two McEliece type cryptosystems and what modifications of the procedures described above are necessary in order to attack them in practice.
The systems we attack are from \cite[Section 5]{bolkema2017variations} and \cite{almeidaBNS21} (note that the latter one is not the authors' current version).
Next, we will describe the McEliece system and the philosophy of how the cryptosystems use convolutional codes in this setting.
The last two subsections are concerned with the practical aspects of the attack on the specific schemes.

\subsection{The McEliece Cryptosystem and Variants based on Convolutional Codes}

The first code-based cryptosystem was introduced by McEliece in his seminal paper \cite{mceliece}.
It can be described as follows.
The parameters $k, n, q, t$ are fixed and all computations are over the finite field $\F_q$.
\begin{itemize}
    \item \textbf{Key Generation}: The secret key is given by a code in the form of a $k \times n$ generator matrix $\mathbf{G}$ together with an efficient decoding algorithm $\phi$ and a random invertible $k \times k$ matrix $\mathbf{U}$ as well as a random $n \times n$ permutation matrix $\mathbf{P}$.
    The public key is given as $\G' = \mathbf{U}\G\mathbf{P}$.
    \item \textbf{Encryption}: To encrypt a message $\mathbf{m}$ of length $k$ the sender computes $\mathbf{r} = \mathbf{m} \G' + \mathbf{e}$, where $\e \in \F_q^n$ is a random error vector of Hamming weight $t$.
    \item \textbf{Decryption}:
    The receiver computes $\phi(\mathbf{r} \mathbf{P}^{-1})$ to recover $\mathbf{m} \mathbf{U}$ and then the original message $\mathbf{m}$.
\end{itemize}

In the cryptosystems in \cite{bolkema2017variations, almeidaBNS21}, the general principle is the same, but the matrices and vectors are not necessarily over $\F_q$ but over $\F_q[z]$.
In particular, the public key $\mathbf{G}'(z)$ is a $k \times n$ matrix with entries in $\F_q[z]$.
Therefore, it generates a convolutional code.
Furthermore, $\mathbf{P}$ might not be a permutation in their schemes.
For the encryption, the authors fix not just the weight, but also the degree of the message $\mathbf{m}(z)$ and the degree of the error vector $\mathbf{e}(z)$, which now also have entries in $\F_q[z]$.
We use the methods we have developed in the previous sections to recover $\mathbf{e}(z)$ from $\mathbf{r}(z) = \mathbf{m}(z) \G'(z) + \mathbf{e}(z)$ applying our ISD framework for convolutional codes with the generator matrix $\G'(z)$.
Since we only need the public key, we do not describe the full details of the key generation and decryption with the private key of the schemes in \cite{bolkema2017variations, almeidaBNS21}.

The first challenge in applying our framework is that the codes constructed in both papers are not delay-free.
This does not pose a major challenge though since we can construct non-catastrophic convolutional codes that contain these codes, see Section \ref{sec: basics}.
Then we use ISD with a fixed amount of iterations with respect to these non-catastrophic codes.

A second problem that we encounter with both systems is the existence of low weight codewords.
To account for this we find codewords of weight $1$ or $2$ with a brute-force search.
Then we add these to all the solutions ISD found and keep only the results that have the weight within specified bounds.
Note that ISD might find several solutions since we are outside of the unique decoding radius.

With these modifications we use a depth first search after each round of ISD as described below.

Let $L_0, M_0$ be the outputs of the first stage of ISD, i.e. $L_0 = (\tilde{\mathbf{e}}_{0,1}, \ldots, \tilde{\mathbf{e}}_{0,n_0} ),$ $M_0 = ( \tilde{\mathbf{m}}_{0,1}, \ldots, \tilde{\mathbf{m}}_{0,n_0})$.
Furthermore, let $t_e$ be the total error weight and $\tilde{\mathbf{G}}^\gamma _i$ the matrices we construct from $\G$ as above.
Moreover let $\mathcal{C}_{low}, \mathcal{M}_{low}$ be the sets of the low weight codewords and corresponding messages, respectively, and let $\mathbf{m^\prime_x}$ correspond to the message for error $\mathbf{x} \in L_j$, and $\mathbf{m^\prime_c}$ to the message for $\mathbf{c} \in \mathcal{C}_{low}$, respectively.

\begin{algorithm}
	\caption{Depth-first sequential ISD} 
	\begin{algorithmic}[1]
        \Require $\text{A generator matrix } \mathbf{G}(z) \text{ of a convolutional code } \mathcal{C},$ $\text{received word }\mathbf{r}(z).$
        \Ensure $\mathbf{e, m} \text{ such that } \mathbf{r}(z)=\mathbf{m}(z)\mathbf{G}(z) +\mathbf{e}(z). $
        \State $j \gets 0, L \gets \{L_0\}, M \gets \{M_0\}$
        \State $\mathbf e \gets \underbrace{(0,0,0,0, 0, \cdots, 0,0,0)}_{s \text{ blocks of length } N}, \mathbf m\gets\underbrace{(0,0,0,0, 0, \cdots, 0,0,0)}_{s \text{ blocks of length } K}$
		\While {$j < s$}
            \While {$L_j \neq \varnothing$}
                \If {$\mathrm{wt}(\mathbf{e})+\mathrm{wt} (L_{j,0})> t_e$}
                \State $L_j \gets L_j -\{L_{j,0}\}, M_j \gets M_j -\{M_{j,0}\}$
                \State \textbf{continue}
                \EndIf
                \State $\mathbf{e}_j \gets L_{j,0}, \mathbf{m}_j \gets M_{j,0}$
                \If {$j=s-1 \And \mathbf{r}(z)-\mathbf{e}(z) \in \mathcal{C}$}
                \State{\Return \textbf{e, m}}
                \Else 
                \State {$L_j \gets L_j -\{L_{j,0}\}, M_j \gets M_j -\{M_{j,0}\}, \mathbf{e}_j \gets 0, \mathbf{m}_j \gets 0$}
                \EndIf
                \State $\mathbf{r}^\prime_{j+1} \gets \mathbf{r}_{j+1} - \sum_{i=0}^{j}\mathbf{m}_i \tilde{\mathbf{G}}^\gamma_{j-i}$
                \State $L_{j+1}, M_{j+1} \gets \operatorname{Prange}(\tilde{\mathbf{G}}^\gamma_0, \mathbf{r}^\prime _{j+1},t+\varepsilon)$
                \If {$L_{j+1} \neq \varnothing$}
                \For {$(\mathbf{x,c}) \in L_{j+1} \times \mathcal{C}_{low}$}
                \If {$ \rm wt(\mathbf{x}+\mathbf{c}) \leq t+\varepsilon$}
                \State $L_{j+1} \gets L_{j+1} \cup \{\mathbf{x} + \mathbf{c}\}$
                \State $M_{j+1} \gets M_{j+1} \cup \{\mathbf{m^\prime_x}+ \mathbf{m^\prime_c}\}$
                \EndIf
                \EndFor
                \State $j \gets j+1$
                \Else
                \State {$L_j \gets L_j -\{L_{j,0}\}, M_j \gets M_j -\{M_{j,0}\}, \mathbf{e}_j \gets 0, \mathbf{m}_j \gets 0$}
                \EndIf    
            \EndWhile
            \State $\mathbf{e}_j \gets 0, \mathbf{m}_j \gets 0$
            \State $j \gets j-1$
            \If {$L_{j+1} = \varnothing$}
            \State $L_j \gets L_j -\{L_{j,0}\}, M_j \gets M_j -\{M_{j,0}\}$
            \EndIf
        \EndWhile
        
	\end{algorithmic} 
\end{algorithm}

Next we describe the experimental results we get for both of the schemes we attack.
We ran our experiments mainly on SageMath 9.4 \cite{sagemath}. Some calculations for the attack in Section \ref{exp2} were done with Magma v2.27/7 \cite{magma}.
The code which contains implementations and the attacks can be found under \url{https://git.math.uzh.ch/projects/2640}.

\subsection{Attack on ``A Variation Based on Spatially Coupled MDPC Codes''} \label{exp1}

We attacked the scheme proposed in \cite{bolkema2017variations}. We took the generator matrix in Example 5.11 for which the authors claim that it gives 80 bits security against generic decoding attacks. The authors propose a $(5,3)$ binary convolutional code. The generator matrix has a memory of $94$. The error is of degree up to $1999$ and has a total weight of $140$. 

When we attacked the system, we took $\gamma=11$, so ISD is done with a $36 \times 60$ generator matrix. The choice of $\gamma$ was motivated by the desire of having, on average, approximately one error in each decoding window. This seems to be a good compromise as for lower $\gamma$'s, the algorithm tends to get lost more often, and for higher $\gamma$'s, the algorithm is slower. At each step, Prange was run 430 times, which guarantees a total success rate of more than $98\%$ according to the formula in Section \ref{wf}. We chose $\varepsilon = 3$ as discussed in Section \ref{sec: epsilon}.
Several seeds for error generation were used and the ones for which the error was outside the range $t + \varepsilon$ were discarded.
From the probability calculations we expect to discard about $22$\% of the seeds. In the experiment we discarded 7 out of 27 errors, which is about 26\%.
For the remaining ones we ran $20$ experiments, the results of which can be seen in Figure \ref{graphic1}.
\begin{figure}[H]
  \centering

\begin{tikzpicture}
    \begin{axis}[
        title={Time (hours) vs. Seed},
        xlabel={Seed},
        ylabel={Time (hours)},
        width=13cm,
        height=8cm,
        symbolic x coords={2,
3,
4,
6,
7,
8,
10,
12,
13,
14,
15,
17,
18,
19,
20,
21,
22,
24,
25,
27},
        xtick=data,
        ymin=3, ymax=10,
        legend pos=north east,
        ymajorgrids=true, 
        enlargelimits=0.15,
        legend style={font=\footnotesize}
    ]
    
    \addplot[
        color=red,
        mark=*,
        mark options={fill=red},
        thick,
        nodes near coords={}, 
        every node near coord/.append style={}
    ] coordinates {
        (2, 5.64833334836695)
(3, 3.12585493584474)
(4, 6.94295182486375)
(6, 3.82750995337963)
(7, 8.61404017620617)
(8, 6.01421779188845)
(10, 5.28673950466845)
(12, 6.16026298913691)
(13, 4.84776269879606)
(14, 4.75675189581182)
(15, 3.94985105223126)
(17, 5.15939919730028)
(18, 7.51064346194267)
(19, 9.84284572899342)
(20, 7.65496225827270)
(21, 4.32993602322208)
(22, 7.55310750226180)
(24, 5.26825166251924)
(25, 8.05503328541915)
(27, 6.99461720665296)

    };

    \end{axis}
\end{tikzpicture}

  \caption{Running time of the experiments for the scheme proposed in \cite{bolkema2017variations}.}
  \label{graphic1}
\end{figure}
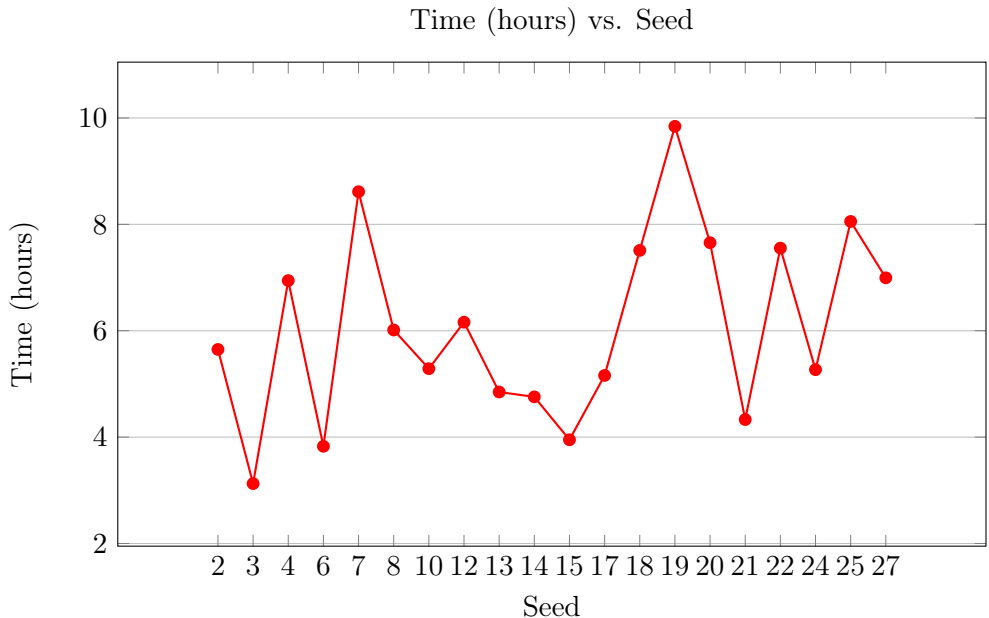
All the experiments correctly recovered the error in less than $10$ hours.
This translates to a bit security of at most $56$ bits. For how we translate between bits and running time, please consult Section \ref{subsec:technicalities}.
Note that we think that our code could be optimized quite a bit and therefore this is an overestimate.
We also ran our experiment with another generator matrix with the key generation as described in \cite{bolkema2017variations}, Section 5.5, for the same parameters to verify our results.
It also finished successfully.
\subsection{Attack on ``Smaller Keys for Code-Based Cryptography: McEliece Cryptosystems with Convolutional Encoders'' 2021} \label{exp2}
We implemented the scheme proposed in \cite{almeidaBNS21} for the parameters for $128$ bit security as given in Example 17 of that paper.
We used several generator matrices and errors to estimate the security of the scheme.
For this scheme we know more about the errors, namely in their example for the $128$ bit security level, the authors use an error with $50$ blocks where the total error weight is $133$ and the sum of the weights of $3$ consecutive blocks is $8$.
This tells us that the blocks have weights $a, b, 3, a, b, 3, \ldots, a, b$ where $a+b = 5$.
Therefore, we take $t=2, \varepsilon = 3$, which would also be a good choice based on the probability calculations. Of course in this case this $\varepsilon$ always works.
For the experiments, we took $\gamma=0$, so $\mathbf{G}_0^\gamma = \mathbf{G}_0$. This means that the size of our block code generated by $\mathbf{G}_0^\gamma$ has size $N=n=62$ and dimension $K=k=30$. At each step, the ISD runs for $700$ iterations, ensuring a total success rate of more than $99\%$.

To give estimates for the running time, we ``cheated'' and checked at each step at which position of the output of the ISD the correct error is. Since the algorithm usually does not run further on a wrong branch than a single node, we expect that adding up the positions of the correct error in the output of the ISD gives us a relatively precise estimate of the total number of ISD iterations needed to recover the correct error.
For the estimates which indicated that we could recover the error quickly, we ran the actual experiments without cheating to verify our estimates.
This was the case for seeds $1$ and $10$.
In both cases, the actual time required was below the estimate.
The results can be seen in Figure \ref{graphic2}. For how we translate between bits and running time, see Section \ref{subsec:technicalities}.

\begin{figure}
\centering
\pgfplotsset{compat=1.17}

\begin{tikzpicture}
    \begin{axis}[
        title={Estimated, Actual bits vs. Seed},
        width=13cm,
        height=8cm,
        xlabel={Seed},
        ylabel={Bits},
        symbolic x coords={1,3,4,6,7,8,9,10,11,13,14,15,16,17,18,19},
        xtick=data,
        ymin=40, ymax=70,
        legend pos=north east,
        ybar,
        bar width=10pt,
        enlargelimits=0.15,
        nodes near coords,
        every node near coord/.append style={xshift=-5pt,yshift=0pt,anchor=east,font=\footnotesize},
        legend style={font=\footnotesize},
        ymajorgrids=true
    ]
    \addplot[draw=black,
        smooth,
        thick,
        mark=*, black,
        nodes near coords={},
        every node near coord/.append style={}
    ] coordinates {
        (1,56.93) (3,59.75) (4,64.09) (6,67.31) (7,64.69) (8,67.88) (9,64.5) (10,51.23) (11,64.55) (13,63.09) (14,60.59) (15,62.84) (16,59.66) (17,65.8) (18,64.03) (19,66.66)
    };
    \addlegendentry{Estimated bits}
    \addplot[
        color=red,
        fill=red
    ] coordinates {
        (1,55.05) (10,50.5)
    };
    \addlegendentry{Actual bits}

    \end{axis}
\end{tikzpicture}
\caption{Estimated bit security for \cite{almeidaBNS21}. For seeds 1 and 10 we also gave the actual running time.}
\label{graphic2}
\end{figure}
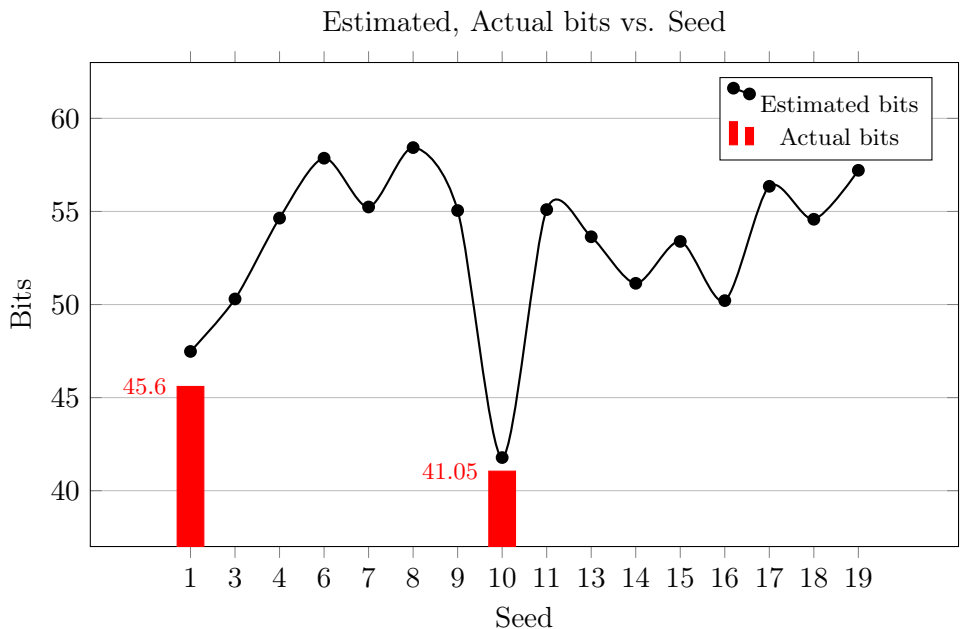
In total we ran 20 experiments, of which 16 successfully finished. In the other cases, the correct error was not found at some step.
In these cases we got only partial recovery due to low weight codewords of weight $3$ and $4$.
We did not include the seeds corresponding to those experiments in the figure above.

\subsubsection{Probability of Getting Lost}

We will justify that, in the attack on \cite{almeidaBNS21}, our algorithm is very unlikely to proceed further than one layer on a wrong branch. Assume the algorithm is correct up to degree $j$, and then, for a given wrong error $\tilde{\mathbf{e}}_{j+1}$ at degree $j+1$, finds a solution $\tilde{\mathbf{e}}_{j+2}$ for degree $j+2$ as well. Let $\mathbf{e}_{j+1}$ and $\mathbf{e}_{j+2}$ be the correct solutions at the respective degrees. Define $\mathbf{c}_{j+1} := \mathbf{e}_{j+1}-\tilde{\mathbf{e}}_{j+1}$ and $\mathbf{c}_{j+2} := \mathbf{e}_{j+2}-\tilde{\mathbf{e}}_{j+2}$. Then:

\begin{itemize}
    \item $\mathbf{c}_{j+1}$ is non-zero,
    \item $(\mathbf{c}_{j+1}, \mathbf{c}_{j+2}) \in \text{rowspan}\begin{pmatrix} \mathbf{G}_0 & \mathbf{G}_1 \\ \mathbf{0} & \mathbf{G}_0 \end{pmatrix},$ 
    \item $\wt(\tilde{\mathbf{e}}_{j+1}) = \wt(\mathbf{e}_{j+1})$ and $\wt(\tilde{\mathbf{e}}_{j+2}) = \wt(\mathbf{e}_{j+2})$.
\end{itemize}
Thus, let $\mathbf{e} \in \mathbb{F}_q^N$ be fixed of weight $t_e$ and $\mathbf{c} \in \mathbb{F}_q^N$ of weight $t_{c}$. We will assume that $N \gg t_e, t_c$, which also holds in practice. We are interested in the probability of $\wt(\mathbf{e} + \mathbf{c}) = \wt(\mathbf{e})$. Let $z$ be the number of elements of $\mathbf{e}$ that become zero after adding $\mathbf{c}$ (the number of ``cancellations''). Notice that to preserve weight we then need to have $z$ non-overlaps, i.e. $z$ elements of the support of $\mathbf{c}$ are outside of the support of $\mathbf{e}$. Notice that if $\wt(\mathbf{e} + \mathbf{c}) = \wt(\mathbf{e})$, then $z \leq \lfloor t_c/2 \rfloor$.  Furthermore, it is also necessary that $z \geq t_c - t_e$ since if $t_c > t_e$, we have at least $t_c - t_e$ entries of $\mathbf{c}$ outside of the support of $\mathbf{e}$, so we also need to cancel at least $t_c - t_e$ positions of $\mathbf{e}$. 
Now, for a fixed number of cancellations, we have a probability of
$$ \frac{\binom{t_e}{z} \binom{N - t_e}{z}(q-1)^z \binom{t_e - z}{t_c - 2 z}(q-2)^{t_c-2z}}{\binom{N}{t_c}(q-1)^{t_c}}$$
that the weight stays the same. Indeed, we have $\binom{t_e}{z}$ choices for entries of $\mathbf{c}$ that cancel entries of $\mathbf{e}$.
For entries of $\mathbf{c}$ outside of the support of $\mathbf{e}$, we choose $z$ from $N-t_e$ positions, and can have $q-1$ entries in each position, which gives us $\binom{N-t_e}{z}(q-1)^z$ possibilities.
Finally, we need to choose $t_c-2z$ positions in the remaining support of $\mathbf{e}$ and have $q-2$ choices for the entries of $\mathbf{c}$ since the entries must be non-zero and also cannot be minus the corresponding entry of $\mathbf{e}$, which is non-zero, giving us $\binom{t_e - z}{t_c - 2z} (q-2)^{t_c-2z}$ possibilities. Let $E_{t_e,t_c}$ be the conditional event that $\wt(\mathbf{e} + \mathbf{c}) = \wt(\mathbf{e})$ for an $\mathbf{e}$ and $\mathbf{c}$ such that $\wt(\mathbf{e}) = t_e$ and $\wt(\mathbf{c}) = t_c$. Summing over all possible numbers of cancellations $z$, we get a probability
$$ \mathbb{P}(E_{t_e,t_c}) = \sum_{z=0}^{\lfloor \frac{t_c}{2} \rfloor} \frac{\binom{t_e}{z}\binom{N-t_e}{z}{\binom{t_e - z}{t - 2z}} (q-2)^{t_c-2z}}{\binom{N}{t_c}(q-1)^{t_c-z}}.$$

We sum $z$ from $0$, as the summands for $z < t_c -t_e$ (when $t_c > t_e$) will be zero anyways. 

Thus, given a vector 
$$(\mathbf{c}_{j+1}, \mathbf{c}_{j+2}) \in \text{rowspan}\begin{pmatrix} \mathbf{G}_0 & \mathbf{G}_1 \\ \mathbf{0} & \mathbf{G}_0 \end{pmatrix},$$
where $\wt(\mathbf{c}_{j+1}) = t_{{c}_{j+1}}$ and $\wt(\mathbf{c}_{j+2}) = t_{{c}_{j+2}}$, and an error vector $(\mathbf{e}_{j+1}, \mathbf{e}_{j+2})$, where $\wt(\mathbf{e}_{j+1}) = t_{{e}_{j+1}}$ and $\wt(\mathbf{e}_{j+2}) = t_{{e}_{j+2}}$, we get that the probability of $ \wt(\mathbf{e}_{j+1}+\mathbf{c}_{j+1}) = \wt(\mathbf{e}_{j+1}) $ and $\wt(\mathbf{e}_{j+2} + \mathbf{c}_{j+2}) = \wt(\mathbf{e}_{j+2})$ is
$$\mathbb{P}(E_{t_{e_{j+1}}, t_{c_{j+1}} }) \cdot \mathbb{P}(E_{t_{e_{j+2}}, t_{c_{j+2}} }).$$

In Table \ref{prob_table}, we have listed $\mathbb{P}(E_{t_{e}, t_{c }})$ for various values of $t_e$ and $t_c$ for the parameters in Example 17 of \cite{almeidaBNS21}.

\begin{table}[h]
\caption{$\mathbb{P}(E_{ t_e,t_c})$ for some values of $t_e$ and $t_c$.}\label{prob_table}
\begin{tabular*}{\textwidth}{@{\extracolsep\fill}lccccc}
\toprule%
& \multicolumn{5}{@{}c@{}}{$t_e$}  \\\cmidrule{2-6}%
$t_c$ & 1 & 2 & 3 & 4 & 5  \\
\midrule
1 & \num{1.58e-2} & \num{3.17e-2} & \num{4.76e-2} & \num{6.35e-2} &\num{7.94e-2}\\
    2 &  \num[round-precision=3,round-mode=figures,
     scientific-notation=true]{0.0005120327700972862}& \num[round-precision=3,round-mode=figures,
     scientific-notation=true]{0.0015194435883917386}& \num[round-precision=3,round-mode=figures,
     scientific-notation=true]{0.0030222324548833577}& \num[round-precision=3,round-mode=figures,
     scientific-notation=true]{0.005020399369572143}& \num[round-precision=3,round-mode=figures,
     scientific-notation=true]{0.007513944332458095}\\
    3  & 0 & \num{4.96e-5} & \num[round-precision=3,round-mode=figures, scientific-notation=true]{0.00017141695394052802}& \num[round-precision=3,round-mode=figures,
     scientific-notation=true]{0.0003882811015135712}& \num[round-precision=3,round-mode=figures,
     scientific-notation=true]{0.0007228804919101439}\\
\bottomrule
\end{tabular*}

\end{table}

\subsubsection{Technical Remarks}\label{subsec:technicalities}
Each of our experiments was run on two UZH I-Math servers.
One of them has a 1 x AMD EPYC\texttrademark\ 7742 2.25GHz CPU with maximum speed of 3.4 GHz, while the other one has a 1 x AMD EPYC\texttrademark\ 7502P 2.5GHz CPU that scales up to 3.35 GHz.
Note that both are AMD EPYC\texttrademark\ 7002 Series CPUS, and specifications for them can be found in \cite{AMD}, \cite{NASAEPYC}. 

Since the AMD EPYC\texttrademark\ 7742 2,25GHz CPU has four integer ALUs and two 256-bit AVX2 we get a maximum throughput of 8.
Each instance used one core, and we did not parallelise our algorithm.
We calculate the upper bound of bit operations per cycle a CPU can handle by taking into account the CPU autoscaling feature (also called turbo) and two 256-bit AVX2 hardware units. Thus, we can translate running time to an upper bound of bits of security. To do this, we first calculate the total number of cycles by multiplying the measured time with the maximum CPU frequency. AVX2 can perform up to 8 integer operations per cycle and an integer operation can consist of up to 64 bit operations. Multiplying all these values and taking the logarithm with base 2 thus gives us an upper bound for the bits of security.
For example, a running time of 10 hours translates to a bit security of at most 
$$\log_2(3.4 \times 10^9\mathrm{ Hz}\times 10 \mathrm{h} \times 3600 \mathrm{s}/\mathrm{h} \times 8 \times 64) \approx 55.8$$
on a 3.4 GHz CPU.

Also note that though we have not implemented our algorithm for multiple cores, it is well-suited for parallelisation, as for the $j+1-$th round of ISD, we only need the values of $e, m$ up to the $j-$th block.
That is to say, each round is only influenced by the nodes above it, not the nodes on the side.
Thus, if we had $n_c$ number of cores, each core could traverse the tree with one of the initial choices for $e_0$. If any core reaches the end of the tree, we could stop there, and if any core goes all the way back to the starting node, having not found the solution in that branch, it could aid any core still actively traversing the tree. This could hypothetically cut down the running time by a factor of at-most $n_c$. We expect parallelisation to be more effective on the attack on \cite{bolkema2017variations}, rather than on \cite{almeidaBNS21}, as in the latter, our algorithm does not seem to get lost as discussed in Section \ref{exp2}.  

\section{Conclusion}

We have developed a framework to iteratively decode convolutional codes with generic decoding methods. We used this framework with information-set decoding to attack two McEliece type cryptosystems where the public key is a convolutional code.
We do not have a theoretical estimate for the running time of our attacks, but experiments suggest that the running time is significantly lower than what the authors estimated for generic decoding methods. In light of our results, we suggest that convolutional code cryptosystems are designed with some defensive adaption such as tail-biting.

There are several challenging open problems for further research.
We do not have theoretical guarantees for how far the algorithm can get lost in the depth first search.
Furthermore the number of low weight codewords that we have to consider in each step of our framework needs further consideration.
Another open point is the choice of the parameter $\gamma$.
While it is clear that a smaller choice gives an exponential speedup on the ISD cost on the block code, the behaviour under depth first search and the number of low weight codewords make a general analysis complicated.
We would like to point out that it seems unlikely that one gets a satisfactory answer to all of these problems considering that for decoding algorithms of convolutional codes it often is hard to get complexity estimates.
For the Viterbi algorithm \cite{viterbi} the complexity grows exponentially in the memory and we are not aware of exact complexity estimates for sequential decoding which do not rely on random coding techniques.
A first step to gain some intuition on how the complexity scales could be to run further experiments with our framework.

This also leaves the status of convolutional code-based cryptosystems to some extent open.
While most of these systems could be broken with considerably less effort than claimed, a full theoretical analysis is difficult.
It is clear that choosing a large enough generator matrix such that ISD for block codes already becomes infeasible thwarts the attack.
On the other hand that would also lead to worse key sizes than for schemes based on linear block codes.

Another open problem, which is more accessible, is to extend our framework to time-varying convolutional codes.
In this case the same principles should apply as long as the diagonal matrices are full rank.

\section*{Acknowledgements}
The authors would like to thank Daniela Portillo del Valle and Joachim Rosenthal for fruitful discussions. The authors would also like to thank Diego Napp and Miguel Beltr\'a Vidal for answering our questions about their system. The authors would further like to thank Yazhuo Zhang for her help with translating experiment time to bit security. The authors would like to thank Violetta Weger and Anna-Lena Horlemann for their encouragement and Carsten Rose and the IT-team at I-Math UZH for their practical help and infinite patience while we tormented the math servers with our experiments. Finally, the authors would like to thank the referees for their constructive feedback and suggested improvements.

\section*{Funding}
This work has been supported in part  by armasuisse Science and Technology (Project number.: CYD-C-2020010), by the Swiss National Science Foundation under SNSF grant number 212865, and by the German research foundation, project number 513811367. The authors have no relevant financial or non-financial interests to disclose.

\end{document}